\documentclass{IEEEtran}
\usepackage{mathpartir,amsmath,mathtools,amsthm}
\usepackage{amsmath,hyperref}
\hypersetup{
    colorlinks,%
    citecolor=blue,%
    filecolor=blue,%
    linkcolor=blue,%
    urlcolor=blue
}
\usepackage{amssymb,verbatim}
\usepackage{fix2col,listings,fancyvrb}
\usepackage{multicol}
\usepackage{graphicx,epsfig}
\usepackage{caption}
\usepackage{stmaryrd}
\usepackage{setspace}
\usepackage{ulem}
\usepackage{newlfont}
\usepackage{epsfig,graphics}
\usepackage{fancybox}
\usepackage{listings}
\usepackage{caption}
\usepackage{graphicx,epsfig}
\usepackage{caption}
\usepackage{makeidx}
\usepackage{algorithm}
\usepackage{algpseudocode}
\usepackage{listings}
\makeindex

\usepackage{algorithm}
\usepackage{algpseudocode}
\usepackage{hyperref}
\hypersetup{
    colorlinks,%
    citecolor=blue,%
    filecolor=blue,%
    linkcolor=blue,%
    urlcolor=blue
}
\usepackage{multicol}
\usepackage{lipsum}

\newtheorem{theorem}{Theorem}

\newtheorem{proposition}[theorem]{Proposition}

\newtheorem{definition}[theorem]{Definition}
\newtheorem{remark}[theorem]{Remark}

\newtheorem{example}[theorem]{Example}

\newcommand{\dom}{\mbox{dom}}

\title{Proof Generation from Delta-Decisions}
\author{Sicun Gao, Soonho Kong, and Edmund M. Clarke\\
{\em Carnegie Mellon University, Pittsburgh, PA, USA 15213}}

\begin{document}
\maketitle
\begin{abstract}
We show how to generate and validate logical proofs of unsatisfiability from delta-complete decision procedures that rely on error-prone numerical algorithms. Solving this problem is important for ensuring correctness of the decision procedures. At the same time, it is a new approach for automated theorem proving over real numbers. We design a first-order calculus, and transform the computational steps of constraint solving into logic proofs, which are then validated using proof-checking algorithms. As an application, we demonstrate how proofs generated from our solver can establish many nonlinear lemmas in the the formal proof of the Kepler Conjecture.
\end{abstract}

\section{Introduction}

Decision solvers for logic formulas over the real numbers play a crucial role in the formal verification of safety-critical embedded systems. For full reliability, decision solvers should provide, besides ``{\sf sat}/{\sf unsat}" answers, certificates of correctness for such answers. For {\sf sat} answers, we can certify by just plugging in a solution of the formula (value assignments for all variables). For {\sf unsat} answers, there is no such witness, and we need mathematical {\em proofs of unsatisfiability} to guarantee correctness. Such proofs are especially important in the framework of $\delta$-complete decision procedures~\cite{DBLP:conf/cade/GaoAC12}, which rely on numerical procedures that are potentially error-prone. For instance, the following is an actual bug we experienced in building our SMT solver dReal~\cite{bug}:
\begin{quote}
With the standard C library implementation {\em eglibc-2.15}, included in the latest {\em Ubuntu 12.10}, the exponential and trigonometric functions contain serious errors. For instance, in upward rounding mode, $\sin(-2.437592) > 10^{53}$. Clearly, this leads to bugs in all constraint solvers using this standard C library.
\end{quote}
Note that when we obtain a proof of unsatisfiability, then the correctness of the result becomes independent from the numerical procedures that were used to obtain them.

Besides certifying correctness of solvers, obtaining such proofs is also important from the perspective of automated theorem proving. Decision solvers can establish mathematical theorems by solving satisfiability of the negation of a theorem, and establish correctness through the absence of counterexamples. Valid proofs of unsatisfiability can be directly used as formal proofs for the theorems. As an approach to automated theorem proving over the real numbers, the scalability can outperform existing symbolic approaches. For instance, Tom Hales' Flyspeck project~\cite{DBLP:conf/dagstuhl/Hales05,DBLP:journals/dcg/HalesHMNOZ10} for the formalization of his proof the Kepler conjecture, requires proving hundreds of nonlinear real inequalities. We will demonstrate that we can automatically generate proofs for many of such formulas. 

It is worth pointing out that after proof generation, proof checking is still a nontrivial problem because of the use of numerical procedures in the computation. Indeed, not all of the {\sf unsat} answers that we have obtained can be proof checked. The challenge lies in validating basic axioms about nonlinear functions over the reals, which can be easily established by numerical algorithms (such as Newton iteration), but not symbolically. Ideally, we need to formalize most of the numerical algorithm in a $\delta$-complete decision procedure to achieve full validation. We regard this as an interesting direction towards bridging the gap between numerical and symbolic methods in  solving formulas over the real numbers. 

We will describe our approach in the following steps:

1. We formalize the ICP algorithm in the framework of Abstract DPLL~\cite{DBLP:journals/jacm/NieuwenhuisOT06}. The similarity between ICP and SAT solving techniques has been explored in existing work~\cite{HySAT}. With this formulation, the branch-and-prune framework is viewed as a transition system with a small set of transition rules. (Section~\ref{formalization})

2. We use a simple first-order proof calculus $\mathbb{D}_A$, relativized to a set $A$ of axioms over the reals, and show how to transform a run of the Abstract ICP to a proof in the system. (Section~\ref{icp})

3. We show how to validate the generated proofs using a stand-alone proof checker implementing simple rules and reliable interval arithmetic. The proof checker interacts with the solver in an abstraction refinement loop to obtain proof trees of sufficient detail (Section~\ref{validate}). In the end, we show experimental results towards the proving nonlinear lemmas in the Flyspeck project, in Section~\ref{kepler}.

\paragraph{Related Work.} Our work is closely related to several lines
of research in the existing literature. For proving formulas with
transcendental functions,
MetiTarski~\cite{DBLP:conf/itp/Paulson12,DBLP:journals/jar/AkbarpourP10,DBLP:conf/aisc/PassmorePM12}
is the leading tool that reduces problems to
polynomials and calls quantifier elimination procedures. Note that in MetiTarski, the polynomial problems are solved using external tools, without producing proofs. For problems with only
polynomials, Bernstein polynomials are used in PVS for formal
proofs~\cite{MN12}. Our approach aims to automatically produce complete formal proofs for formulas with transcendental functions. The iSAT solver~\cite{HySAT} also contains strategies for certifying
their answers in a different framework~\cite{DBLP:conf/ddecs/KupferschmidBTF11}.
There are now several SMT
solvers~\cite{DBLP:conf/cade/JovanovicM12,DBLP:conf/mkm/PassmoreJ09} for formulas with nonlinear polynomials over
the reals based on CAD with no proof-producing capacities, but a proof-producing
algorithm is possible, as sketched in~\cite{DBLP:conf/cade/McLaughlinH05}. Proofs for
correctness in general SMT solvers have been well studied, for instance in~\cite{DBLP:journals/fmsd/StumpORHT13}, which allows us to focus on the nonlinear theory solver in our framework.

\section{A Formalization of Interval Constraint Propagation}\label{formalization}

Interval Constraint Propagation (ICP)~\cite{handbookICP} finds
solutions of real constraints using the ``branch-and-prune'' method, combining
interval arithmetic and constraint propagation. The idea is to use interval
extensions of functions to ``prune'' out sets of points that are not in the
solution set and ``branch'' on intervals when
such pruning can not be done, recursively until a small enough box
that may contain a solution is found or inconsistency is observed. A high-level description of the decision version of ICP is given in Algorithm~\ref{icpalgo}~\cite{handbookICP,DBLP:conf/cade/GaoAC12}.
\begin{algorithm}\label{algo1}
\caption{ICP($f_1,...,f_m, B_0 = I_1^0\times\cdots\times I_n^0, \delta$)}\label{icpalgo}
\begin{algorithmic}[1]
\Statex
    \State $S \gets B_0$
    \While{$S \neq \emptyset$}
        \State $B \gets S.\mathrm{pop}()$
        \While{$\exists 1 \leq i \leq m, B \neq_{\delta} \mathrm{Prune}(B,f_i)$}
        \State $B \gets \mathrm{Prune}(B, f_i)$
        \EndWhile
        \If{$B\neq \emptyset$}
            \If{$\exists 1\leq i\leq m, |\sharp f_i(B)|\geq \delta$}
                \State $\{B_1,B_2\} \gets \mathrm{Branch}(B, i)$
                \State $S.\mathrm{push}(\{B_1,B_2\})$
            \Else
                \State \Return {\sf sat}
            \EndIf
        \EndIf
    \EndWhile
    \State \Return {\sf unsat}
\end{algorithmic}
\end{algorithm}

Our task now is to formalize ICP algorithms so that we can extract symbolic proofs from its computation processes. Similar to Abstract DPLL, we represent ICP
as a transition system, whose states consist of interval assignments and the
real constraints to be solved. An interval $I$ is any connected subset of
$\mathbb{R}$ and we write $\mathbb{IR}$ to denote the set of all the intervals.
We first formalize how ICP maintains interval assignments to a set of variables
as follows:
\begin{definition}[Interval Assignment Sequence]
Let $x_1,...,x_n$ be real variables. An {\em interval assignment sequence} over
$\vec x$ is a sequence $(s_1,...,s_m)$, where
\begin{multline*}
s_i\in \{(x_i\in I_j): 1\leq i\leq n, I_j\in
\mathbb{IR}\}\\
\cup\{(x_i\in I_j)^d: 1\leq i\leq n, I_j\in
\mathbb{IR}\}.
\end{multline*}
We write $(S_1, S_2)$ to denote the concatenation of two sequences $S_1$ and
$S_2$. The parentheses can be omitted when appropriate.
\end{definition}
It will be clear later that when we write $(x\in I)^d$, it means an arbitrary
choice on the value of $x$ (called a d-assignment), which is consequently a backtrack point.
\begin{remark}
ICP can maintain unions of intervals for variables. In principle this is not
needed if we only consider the decision problem, which only searches for
one solution and the components of a union can be tested sequentially. So we
assume that only connected subsets of values are used here.
\end{remark}

\begin{definition}[Box Domain]
Let $S$ be an interval assignment sequence over variables $x_1,...,x_n$. The
{\em box
domain} associated with $S$ is defined by
$$\beta(S) = I_1\times\cdots \times I_n,$$ 
where $I_i = \bigcap\{ I:
(x_i\in
I)\mbox{ or } (x_i\in I)^d \mbox{ occurs in } S\}.$
Also, we write $\beta(S)_i$ to denote $I_i$.
\end{definition}

\begin{definition}[ICP Transitions]~\label{transitions} Let $\vec x =
(x_1,...,x_n)$ be a vector of real variables. We write $c(\vec x)$ to denote a constraint over $\mathbb{R}^n$, and
$S$ an interval assignment sequence over $\vec x$. Let $S\parallel c$ be the
current state. We will always write $\beta(S_i) = I_i$ to denote the current
interval assignment on variable $x_i$. We now define the following transition
rules from $S\parallel c$ to another state.
\paragraph{(Pruning)} Let $I_i^1$ be a subset of $I_i$ such that $\forall \vec
a\in \beta(S,x_i\in I_i^1)$, $c(\vec a)$ is false. Then, if we let $I_i^2$ be an
interval satisfying $I_i\subseteq I_i^1 \cup I_i^2$, then
\begin{eqnarray*}
S\parallel c &\stackrel{p}{\Longrightarrow}& S, (x_i\in I_i^2)\parallel c
\end{eqnarray*}
is called a pruning step.
\paragraph{(Branching)}Let $I_i^1$ be a subset of $I_i$. Then
\begin{eqnarray*}
S\parallel c &\stackrel{br}{\Longrightarrow}& S, (x_i\in I_i^1)^d \parallel c,
\end{eqnarray*}
is called a branching step.
   
\paragraph{(Backtracking)} Let $I_i^1$ be a subset of $I_i$, such that $\forall
\vec a\in \beta(S,x_i\in I_i^1, S')$, $c(\vec a)$ is false. Let $I_i^2$ be an
interval such that $I\subseteq I_i^1\cup I_i^2$. If in addition, $S'$ does not
contain any $d$-assignment (of the form $(x\in I)^d$), then we can make
a transition
\begin{eqnarray*}
S, (x_i\in I_i^1)^d, S'\parallel c \stackrel{bt}{\Longrightarrow}& S, (x_i\in
 I_i^2) \parallel c,
\end{eqnarray*} 
which is called a backtracking step.

\paragraph{(Failure)} Suppose $\forall \vec a\in \beta(S)$, $c(\vec a)$ is
false, and there is no $d$-assignment in $S$. Then we can make the transition
\begin{eqnarray*}
S\parallel c \stackrel{f}{\Longrightarrow} \emptyset\parallel c
\end{eqnarray*}
which is called a failure step.
\end{definition}
\begin{definition}[Abstract ICP]
An $n$-dimensional ICP framework is a transition system
$$\langle \mathbb{IR}^n, \mathcal{S}, \mathcal{C}, \Longrightarrow,
\varepsilon\rangle$$
where $\mathcal{S}$ is the set of all interval assignment sequences over
$\mathbb{IR}^n$, and $\mathcal{C}$ is any set of constraints over
$\mathbb{R}^n$. A state is an element in $\mathcal{S}\parallel \mathcal{C}$. The
transition rules $\Longrightarrow: \mathcal{S}\times \mathcal{C}\rightarrow
\mathcal{S} \times \mathcal{C}$ are as specified in
Definition~\ref{transitions}. $\varepsilon\in \mathbb{Q}^+$ is an
error bound. A {\em run} of ICP is any sequence
$$S_1\parallel c, ... , S_k\parallel c,$$
where either $S_k$ is $\emptyset$, or $S_k\neq \emptyset$ and
$||\beta(S_k)||<\varepsilon$.
\end{definition}
\begin{remark}
We have defined ICP in a general way, without enforcing conditions
on the pruning operators, such as well-definedness. Thus, many
invalid ICP runs can be generated. In this way, we treat ICP as a proof
searching algorithm, and rely on the proof checkers to determine the correctness
of an ICP run. In practice, of course, only ``correct" ICP algorithms can
provide proofs that can always be validated.
\end{remark}
\begin{example}~\label{icp-example}
Consider a constraint $c(x,y) : y=x \wedge y = x^2$, and $x\in [1.5,2]$ and
$y\in [1,2]$ are the initial interval assignment. A possible ICP run is:
\begin{eqnarray*}
& &x\in [1.5,2], y\in [1,2]\parallel c \\
&\stackrel{br}{\Longrightarrow}& x\in [1.5,2],
y\in [1,2], (x\in [1.7, 2])^d\parallel c\\
&\stackrel{bt}{\Longrightarrow}& x\in [1.5,2], y\in [1,2], x\in [1.5, 1.7]\parallel c\\
& &\hspace{-.5cm} \mbox{  (backtracking, since $\forall\vec a\in [1.7,2]\times [1,2]$,
$c(\vec a)$ is false,}\\
& & \ \ \mbox{and $[1.5,2]\subseteq[1.5,1.7]\cup [1.5, 2]$ for
$x$)}\\
&\stackrel{p}{\Longrightarrow}& x\in [1.5,2], y\in [1,2], x\in [1.5, 1.7], x\in
[1.5, 1.6]\parallel c\\
& & \mbox{  (pruning, since $\forall \vec a\in[1.6,1.7]\times [1, 2]$, $c(\vec
a)$ is false)}\\
&\stackrel{p}{\Longrightarrow}& x\in [1.5,2], y\in [1,2], x\in [1.5, 1.7],\\
& & \hspace{4cm}x\in [1.5, 1.6], x\in \emptyset\parallel c\\
& & \mbox{  (pruning, since $\forall \vec a\in[1.5,1.6]\times [1, 2]$, $c(\vec
a)$
is false)}\\
&\stackrel{f}{\Longrightarrow}& \emptyset||c\ \mbox{ (since
$\forall \vec a\in \emptyset \times [1, 2], c(\vec a)$ is false.)}
\end{eqnarray*}
\end{example}

\section{Extracting Proofs from ICP Runs}\label{icp}

\subsection{First-Order Proofs of Unsatisfiability}

We focus on the proof the unsatisfiability of conjunctions of theory
atoms in the DPLL(T) framework, i.e., formulas of the form
$$\exists^{I_1} x_1\cdots \exists^{I_n} x_n.\; \bigwedge_{i=1}^m
f_i(x_1,...,x_n)\sim 0$$
where $\sim \in \{=,\neq, >, \geq, <, \leq\}$. It is clear that once such
proofs are obtained, the proof of unsatisfiability of Boolean combinations of
the theory atoms can be obtained, by simply plugging them in the high level
resolution proof. Also, it is important to note that the ICP algorithm solves
{\em systems} of constraints, and it regards the conjunction $\bigwedge_{i=1}^m
f_i(x_1,...,x_n)\sim 0$ as one constraint $c(x_1,...,x_n)$. Consequently, our
task is now reduced to obtaining proofs for the validity of formulas of the form
$\forall x_1 \cdots \forall x_n.(x_1\in I_1\wedge \cdots \wedge x_n\in I_n
)\rightarrow \neg c(\vec x)$, from the failure of ICP
search for a solution to the original SMT formula $\exists \vec x. \vec x\in
\vec I\wedge c(\vec x)$.

We will construct a simple first-order proof calculus, and show how to
transform ICP runs into proofs in the system.

Again, we consider formulas in a signature $\mathcal{L}_F = \langle <,
\mathcal{F} \rangle$, where constants are considered as 0-ary functions in
$\mathcal{F}$. When we write $x\in I$, where $I$ denotes an
interval, it is regarded as an abbreviation for their
equivalent $\mathcal{L}_\mathcal{F}$-formula. Note that this means that $I$ only
uses $\mathcal{L}_{\mathcal{F}}$-terms as end-points. Also, as mentioned above,
$c(\vec x)$ abbreviates a conjunction of atomic formulas. We also allow the use
of vectors in the formulas, writing $\vec x\in \vec I$ to denote $\bigwedge_i
x_i\in I_i$.

\begin{definition}[System $\mathbb{D}_A$]
We define $\mathbb{D}_A$ to be the first-order proof system consisting of only the
following two rules:
  \begin{mathpar}
    \inferrule{
      \forall \vec x (\psi \rightarrow \varphi)
      \and
    \forall \vec x (\psi' \rightarrow \varphi)
    }
  {
  \forall \vec x ( \psi\vee \psi' \rightarrow \varphi)
  }{\ \ \vee I}\and
  \inferrule{
  \forall {\vec x} (\psi \rightarrow \varphi)
    \and
  \forall {\vec x} (\psi' \rightarrow \psi)
  }
  {
  \forall {\vec x} (\psi' \rightarrow \varphi)
  }{\ \ \forall\mbox{MP}}
  \end{mathpar}
and a set $A$ of axioms of the following two types:
\paragraph{Interval Axioms}
\begin{mathpar}
\inferrule{ }{\forall \vec x (\vec x\in \vec I \rightarrow \vec x\in \vec I_1
\vee \vec x\in \vec I_2 )}{\ \mbox{IA}}
\end{mathpar}
\paragraph{Constraint Axioms}
\begin{mathpar}
\inferrule{ }{\forall \vec x ( \vec x\in\vec I \rightarrow c(\vec x))}{\
\mbox{CA}}
\end{mathpar}
\end{definition}
Derivations in $\mathbb{D}_A$ are as standardly defined, as natural deductions
following these rules. Clearly, the two first-order rules are valid. Thus,
if all the axioms in $A$ are valid, then the system only produces valid
formulas over $\mathbb{R}$.
\begin{proposition}[Soundness] If $\mathbb{D}_A\vdash \varphi$ and
$\mathbb{R}\models \bigwedge A$, then $\mathbb{R}\models \varphi$.
\end{proposition}
\begin{remark}
 Clearly, the constraint axioms are the most nontrivial part. They are the
basic facts of real functions that a numerical procedure relies on, usually
concerning the range of functions within a small interval. The interval axioms
are sometimes not trivial either (for instance, compare intervals ending with
$e^{\pi}$ and $\pi^e$ respectively).  Proof-checking involves validation of
these axioms, which we discuss in Section~\ref{validate}.
\end{remark}

We now describe the construction of proof trees from ICP runs, which will be
represented as labeled binary trees. A labeled binary tree is defined as a
tuple $T =
\langle V, V_L, \Sigma, \delta, \sigma\rangle$. Here, $V = \{v_0, ..., v_k\}$, is a finite set of nodes, where $v_0\in V$ always denotes the
root
node. $V_L$ is the set of leaf nodes in $V$. $\Sigma$ is a set of labels, which in our case is the set of
$\mathcal{L}_\mathcal{F}$-formulas. $\delta:\subseteq V\times \{l,r\}
\rightarrow V$ is a
partial mapping from a node to its descendant nodes, where $\delta(v, l)$ and
$\delta(v, r)$ denote the left and right descendant nodes, respectively.
$\sigma:\subseteq V\rightarrow \Sigma$ is a labeling function that maps each node
$v\in V$ to a formula $\sigma(v) \in \Sigma$. In addition, the edges in the
tree can be labeled as well, through a function $\tau: V\times V\rightarrow \Omega$
where $\Omega$ is a set of edge-labels.

\subsubsection{Tree Generation} Let an ICP run be
$$S_0\parallel c\stackrel{t_1}{\Longrightarrow}\cdots
\stackrel{t_m}{\Longrightarrow} S_m\parallel c,$$
such that the ending transition $t_m$ is a failure step, i.e., $S_m=\emptyset$.
We now define the procedure by defining the functions $\delta$ and
$V_L$ through induction on $s_i$. The edges can be labeled naturally with
$\Omega$ = \{{$\vee$I}, $\forall$M, IA, CA\}.

\paragraph{Case $i= 0$.} We label the root node $v_0$ by
$$\sigma(v_0) := \forall \vec x( \vec x\in \beta(S_0) \rightarrow \neg c).$$
Let $V_L^0= \{v_0\}$ denote the current
collection of leaf nodes. Note that this formula is the negation of the input
SMT formula.


\paragraph{Case $i = k+1$ ($1< k \leq m$). }
Suppose $V_L^k$ and $\sigma$ have been defined for $s_1, ...,s_k$. Write
$s_k = S_k
\parallel c$ and $s_{k+1} = S_{k+1} \parallel c$. Now we split
the cases on the type of the step $t$ from $s_k$ to $s_{k+1}$ as follows.
Again, we use the convention that $\beta(S)_i = I_i$ denotes the current
interval assignment on a variable $x_i$.
\paragraph{(Pruning Case)} Suppose $s_k\Longrightarrow s_{k+1}$ is a
pruning step. That is,
$$S_k\parallel c \stackrel{p}{\Longrightarrow} S_k, (x_i\in I_i^2)\parallel c,$$
where $I_i\subseteq I_i^1\cup I_i^2$ and $\forall \vec
a\in \beta(S_k,x_i\in I_i^1)$, $c(\vec a)$ is false. If we write
$$\vec I_1 = \beta(S_k, (x_i\in I_i^1)), \vec I_2 = \beta(S_k, (x_i\in
I_i^2)), \mbox{ and } \vec I= \beta(S_k),$$ then this step corresponds to
the sub-tree as shown in Fig. 1, Case A.  

\begin{figure*}[h!]
A. Pruning Case:
{\small
\begin{mathpar}
\inferrule{\inferrule{
  \inferrule{\vdots}{\forall \vec x (\vec x\in \vec I_2 \rightarrow \neg c)}
    \and
    \inferrule{\ }
    {
      \forall \vec x (\vec x \in \vec I_1 \rightarrow\neg c)
      }\mbox{CA}
   }
  {
  \forall x ( (\vec x\in \vec I_1\vee \vec x \in \vec I_2) \rightarrow \neg c)
  }\mbox{$\vee$I}
  \and
  {
  \inferrule{\ }{\forall x ( x\in I_i\rightarrow(x \in
I_i^1 \vee x\in I_i^2))}\mbox{IA}
  }
  }
{
\forall \vec x (\vec x\in\vec I \rightarrow \neg c)
}\mbox{$\forall$MP}
\end{mathpar}}
B. Branching Case:
{\small
\begin{mathpar}
\inferrule{\inferrule{
  \inferrule{\vdots}{\forall \vec x (\vec x\in \vec I_1 \rightarrow \neg c)}
    \and
    \inferrule{\vdots }
    {
      \forall \vec x (\vec x \in \vec I_2 \rightarrow\neg c)
      }
   }
  {
  \forall x (\vec x\in \vec I_1\vee \vec x \in \vec I_2 \rightarrow \neg c)
  }\mbox{$\vee$I}
  \and
  {
  \inferrule{\ }{\forall x ( x\in I_i\rightarrow(x \in
I_i^1 \vee x\in I_i^2))}\mbox{IA}
  }
  }
{
\forall \vec x (\vec x\in\vec I \rightarrow \neg c)
}\mbox{$\forall$MP}
   \end{mathpar}
}
C. Backtracking Case:
{\small
\begin{mathpar}
\inferrule{
  \inferrule{
    \inferrule*{
      \inferrule*[vdots=1.5em, Right=\mbox{CA}]{ }{
        \forall \vec x
        (\vec x\in \beta(S_{k'},(x\in I_i^1)^d,S')\rightarrow \neg c)
      }
    }
    {\forall \vec x (\vec x\in \vec I_1
      \rightarrow \neg c)}
    \and
    \inferrule{
      \vdots
      }
    {
      \forall \vec x (\vec x \in \vec I_2 \rightarrow\neg c)
    }
   }
  {
  \forall x (\vec x\in \vec I_1\vee \vec x \in \vec I_2 \rightarrow \neg c)
  }\mbox{$\vee$I}
  \and
  {
\vdots\ \ \
  }
  }
{
\forall \vec x (\vec x\in\vec I \rightarrow \neg c)
}\mbox{$\forall$MP}
   \end{mathpar}
}
\caption{Proof Trees}\label{trees}
\end{figure*}

Formally, the sub-tree is added as follows. Let $v\in V_L^k$ be an
existing
leaf node that is labeled by the formula corresponding to $S_k\parallel c$;
namely,
$$\sigma(v) = \forall \vec x (\vec x\in\vec I \rightarrow \neg c).$$ (We
will inductively prove that such a node exists.) We then define
\begin{eqnarray*}
\delta(v, l) &=& v_{k+1}^1, \sigma(v_{k+1}^1) = \forall \vec x
( (\vec x \in \vec I_1 \vee \vec x \in \vec I_2) \rightarrow \neg c); \\
\delta(v, r) &=& v_{k+1}^2, \sigma(v_{k+1}^2) = \forall \vec x ( \vec x\in \vec
I\rightarrow( \vec x \in \vec I_1 \vee \vec x\in \vec I_2));\\
 \delta(v_{k+1}^1, l) &=& v_{k+1}^3, \sigma(v_{k+1}^3) = \forall \vec x (\vec x
\in \vec I_2 \rightarrow \neg c)\\
 \delta(v_{k+1}^1, r) &=& v_{k+1}^4, \sigma(v_{k+1}^4) = \forall \vec x (\vec x
\in \vec I_1 \rightarrow\neg c)
  \end{eqnarray*}
and set $V_L^{k+1} = (V_L^k \setminus\{v\})\cup \{v_{k+1}^3\}$.

\paragraph{(Branching Case)}Suppose $s_k\Longrightarrow s_{k+1}$ is a
branching step. That is,
\begin{eqnarray*}
S_k\parallel c &\stackrel{br}{\Longrightarrow}& S_k, (x_i\in I_i^1)^d \parallel
c,
\end{eqnarray*}
under the condition that $I_i^1\subseteq I_i$. If we write
$$\vec I_1 = \beta(S, (x_i\in I_i^1)), \vec I_2 = \beta(S, (x_i\in
I_i^2)), \mbox{ and } \vec I= \beta(S),$$ where $I\subseteq I_i^1\cup I_2$, then
this step corresponds to the sub-tree as shown in Fig. 1, Case B. Formally it is defined as follows. Again, let $v\in V_L^k$
be a leaf node such that $\sigma(v) = \forall \vec x (\vec x\in\vec I
\rightarrow \neg c).$ We then define
\begin{eqnarray*}
\delta(v, l) &=& v_{k+1}^1, \sigma(v_{k+1}^1) = \forall \vec x
( \vec x \in \vec I_1 \vee \vec x \in \vec I_2 \rightarrow \neg c); \\
\delta(v, r) &=& v_{k+1}^2, \sigma(v_{k+1}^2) = \forall \vec x ( \vec x\in \vec
I\rightarrow(\vec x \in \vec I_1 \vee \vec x\in \vec I_2));\\
 \delta(v_{k+1}^1, l) &=& v_{k+1}^3, \sigma(v_{k+1}^3) = \forall \vec x (\vec
x\in
\vec I_1 \rightarrow \neg c)\\
 \delta(v_{k+1}^1, r) &=& v_{k+1}^4, \sigma(v_{k+1}^4) = \forall \vec x (\vec x
\in \vec I_2 \rightarrow\neg c)
  \end{eqnarray*}
and set $V_L^{k+1} = (V_L^k\setminus \{v\}) \cup \{v_{k+1}^3, v_{k+1}^4 \}$.

\paragraph{(Backtracking Case)}Suppose $s_k\Longrightarrow s_{k+1}$ is
a branching step. That is,
\begin{eqnarray*}
S_{k'} , (x_i\in I_i^1)^d, S'\parallel c &\stackrel{bt}{\Longrightarrow}& S_{k'},
(x_i\in I_i^2 ) \parallel c,
\end{eqnarray*}
when $\forall a\in \beta(S, (x_i\in I_i^1)^d, S')$, $c(\vec
a)$ is false, and $I_i\subseteq I_i^2\cup I_i^1$, where $I_i =
\beta(S_{k'})_i$. $S_{k'}$ is a previous interval assignment
sequence, with $k'<k$. If we write
$$\vec I_1 = \beta(S, (x_i\in I_i^1)), \vec I_2 = \beta(S, (x_i\in
I_i^2), \mbox{ and } \vec I= \beta(S_{k'}),$$ then this step corresponds to the sub-tree as shown in Fig. 1, Case C. Formally, we simply set $V_L^{k+1} = V_L^{k}$, and do not update the labeling.

\paragraph{(Fail Case)} Suppose it is a failure step. That is,
\begin{eqnarray*}
S\parallel c &\stackrel{f}{\Longrightarrow}& \emptyset \parallel c
\end{eqnarray*}
when $\forall \vec a\in \beta(S)$, $c(\vec a)$ is false and
 $S$ has no $d$-assignments. Let $\vec I
=\beta(S)$. This step corresponds to
{\small\begin{mathpar}
 \inferrule{\ }{\forall \vec x ( \vec x\in \vec I) \rightarrow \neg
c}\mbox{FA}
\end{mathpar}}We set $V_L^{k+1}=V_L^k\setminus\{v\}$ and do not update $\sigma$.

\paragraph{Complete tree.} In all, after all the steps in the ICP run are
followed, the tree that we construct is $ T = \langle V, V_L^m, \Sigma, \delta, \sigma
\rangle$. The axiom set is given by $$A = \{\sigma(v): v\in V_L^m\}.$$

It is easy to see that $T$ is a valid proof tree in $\mathbb{D}_A$:

\begin{proposition}\label{successful_tree}
For every ICP run ending with $\emptyset\parallel c$, the tree construction procedure above produces a valid natural deduction tree in $\mathbb{D}_A$. The size of the proofs is linear in the computation steps.
\end{proposition}

\begin{proof}
It is clear that each proof step, as represented by the subtree created in each case, is a valid natural deduction step in $\mathbb{D}_A$. We only need to show that the tree can be constructed. For this, we need to show that for each step $S_k\parallel c\stackrel{t}{\Longrightarrow} S_{k+1}\parallel c$, where $S_{k+1}$ is not $\emptyset$, it is always the case that $S_k\parallel c$ labels a leaf node in the tree constructed so far. When $k=0$, this is the case since $V_L^0 = \{v_0\}$. Now suppose $S_k\parallel c$ labels a leaf node. If $t$ is a pruning step, then $\forall \vec x(\vec x\in \vec I_2\rightarrow \neg c)$ labels $v^3_{k+1}$, which is added in $V_L^{k+1}$. The same applies to the other branching and backtracking. Finally, the step $S_{m-1}||c\Longrightarrow \emptyset||c$ corresponds to closing the last leaf labelled by $\forall \vec x(\vec x \in \vec I\rightarrow \neg c)$.
\end{proof}
Again, once the proof tree is constructed, the details of
the ICP algorithm no longer matters.  The only rules involved are the two
first-order rules in $\mathbb{D}_{A}$. Following relative soundness of the
system, to establish validity of the formula, now we only need to validate the
axiom set $A$.

\section{Validating the Proofs}\label{validate}

\subsection{Validating the Axioms}

There are two types of axioms that we allow in the proofs constructed from ICP runs: interval axioms and constraint axioms. To validate such axioms, we still need numerical computations. The difference is that the proof checker only needs to implement simple interval computation that can be validated through stand-alone arbitrary-precision or symbolic computation. Note that the validation of the axioms can fail when the solver correctly returns {\sf unsat}, if the solver uses complex numerical heuristics that can not be verified by reliable numerical computation. In practice, we ensure the correctness of the proof checker first, and use an abstraction refinement loop that allows the proof checker to ask for more detailed proofs from the solver.

The interval axioms do not contain any real functions, and are of the form
$\forall \vec x(x\in I_1\vee x\in I_2\rightarrow x\in I)$. We only need to
check that $I$ is a subset of $I_1\cup I_2$ by comparing the end points of the
intervals, which is an easy numerical task.

The constraint axioms are of the form $\forall x (\vec x\in \vec I \rightarrow
c(\vec x))$, and can only be verified by considering the functions that
occur in $c$. Although they are of the same form as the formulas we solve,
these axioms should contain evident properties of the functions involved,
usually on small intervals. Such facts can be verified using reliable interval
computations, for instance as follows.
\begin{definition}[Interval Extensions]
Let $f: \mathbb{R}^n\rightarrow \mathbb{R}$ be a real function. An interval
function $F: \mathbb{IR}^n \rightarrow \mathbb{IR}$ is a function that
satisfies:
$$\forall I\in \dom(F), \{f(x): x\in I\}\subseteq F(I).$$
\end{definition}
A simple example of interval extensions is the {\em natural interval
extension} for
arithmetic operations, based on computations of functions on the end points of
intervals. It is obvious that:
\begin{proposition}
Let $F$ be an interval extension of $f$, and $I\subseteq \dom(f)$. If
$F(I)\subseteq A$, then $\forall x (x\in I \rightarrow f(x)\in A)$.
\end{proposition}
Thus, the axioms are validated if we can verify that they are consistent with
all the interval extensions.
\begin{example}
The second pruning step in Example~\ref{icp-example} generates an axiom
$$\forall
x \forall y ( x\in [1.7, 2] \wedge y\in [1, 2] \rightarrow \neg (y=
x^2)\vee \neg (y=x))$$
This can be easily validated through the natural interval extension of
$(y-x^2)$, which is $[1,2]-[1.7,2]^2 = [-3, -0.89]$ and does not contain $0$.
\end{example}

\subsection{Taylor Proofs}
Suppose we want to verify the inequality $f(x_1,...,x_n) > 0$ on a domain $\vec x\in D = I_1\times \cdots \times I_n.$ Using the multivariate mean value theorem, we have that for any $\vec a, \vec b\in D$
$$f(\vec b) - f(\vec a) = \nabla f(\xi)\cdot (\vec b - \vec a) = \sum_i \frac{\partial f}{\partial x_i}(\xi)\cdot (b_i-a_i)$$
for some $\xi\in D$. Thus, we can bound $f(x)$ on $D$ by computing the interval bound on the function
$$f(\vec a)+\sum_i \bigg(\sharp\bigg(\frac{\partial f}{\partial x_i}\bigg)(D)\bigg)\cdot D|_{x_i}$$
where $\sharp(\cdot)$ denotes interval extension, and $f(\vec a)$ is on the boundary of $D$.

\begin{example}
$f(x_1, x_2) = x_1^2 + x_2^2$ on domain $(x_1,x_2)\in [0,1]\times[0,1]$. We have $\partial f/\partial x_1 = 2x_1\in [0,2]$ and $\partial f/\partial x_2\in [0,2]$. Thus
$$f(\vec x)\in \sum_{i=1,2} [0,2]\cdot (1-0) + 0 = [0,4].$$
\end{example}

\subsection{The Branch and Prove Loop}

In practice, ICP usually implements complicated heuristics that are more
powerful than what can be verified through direct interval arithmetic. A practical approach first is to use an
abstraction refinement loop that allows the proof checker to ask the
solver for proof traces of the right amount of details. We sketch the procedures in Algorithm~\ref{b1} and Algorithm~\ref{b2}.

\begin{algorithm}
  \centering
  \caption{ProofCheck}
  \label{alg:proofcheck}
  \begin{algorithmic}[1]
    \Procedure{ProofCheck}{$p$, $\delta$}
        \If {$\Call{match}{p, \mathrm{Axiom}(\forall x (\vec x\in \vec I \rightarrow c(\vec x)))}$}
            \If {$\sharp c(\vec I)$} \Comment use interval arithmetic, taylor extension..
                \State \Return $\emptyset$
            \Else
                \State $(\vec I_1, \vec I_2) \gets \mathrm{Split}(\vec I)$
                \State $(\delta_1, \delta_2) \gets (\min( \delta, \frac{1}{4}||I_1||), \min( \delta, \frac{1}{4}||I_2||))$
                \State \Return $\{
                (\forall x (\vec x\in \vec I_1 \rightarrow c(\vec x)), \delta_1),
                (\forall x (\vec x\in \vec I_2 \rightarrow c(\vec x)), \delta_2)
                \}$
            \EndIf
        \ElsIf {$\Call{match}{p, \mathrm{Branch}(p_1, p_2, \vec I)}$}
            \State $U_1 \gets \Call{ProofCheck}{p_1, \delta}$
            \State $U_2 \gets \Call{ProofCheck}{p_2, \delta}$
            \If{$\vec I \not \subseteq (\dom(p_1) \cup \dom(p_2))$}
                \State \Return Error
            \Else
                \State \Return $U_1 \cup U_2$
            \EndIf
        \EndIf
    \EndProcedure
  \end{algorithmic}\label{b1}
\end{algorithm}

\begin{algorithm}
  \centering
  \caption{Branch-and-Prove}
  \label{alg:branch-and-prove}
  \begin{algorithmic}[1]
    \Procedure{Branch-and-Prove}{$p$, $\delta$}
        \State $U \gets \Call{ProofCheck}{p, \delta}$
        \If {$U \not = \emptyset$}
            \ForAll {$(a, \delta') \in U$}
                \State $p' \gets \Call{Solve}{a, \delta'}$
                \State \Call{Branch-and-Prove}{$p'$, $\delta'$}
            \EndFor
        \EndIf
    \EndProcedure
  \end{algorithmic}\label{b2}
\end{algorithm}

When we fail to prove an axiom through simple interval arithmetic, the
proof checker generates new subproblems that are returned to the
solver. At this stage, the axioms become the new theorems to be
proved. This is an abstraction refinement procedure.
Algorithm~\ref{alg:branch-and-prove} illustrates the procedure. By
executing the loop, we may obtain proof trees that contain more and more
detailed steps. There are two ways that the prover can generate the
subproblems, branching on a variable in the formula or using a smaller
$\delta$. Note that under the condition that the pruning operators in
the solver is well-defined, both procedures never change the {\sf
  unsat} result. The branching may give exponentially many new
problems; while the $\delta$-change does not give new problems, but
may exponentially slow down the solver in each round. In practice we
observe that such a refinement loop is very useful, as we will show in
the experiments.

\section{Experiments}\label{kepler}
We implemented the proof generation capacity into our open-source solver {\sf
  dReal}\footnote{\url{http://dreal.cs.cmu.edu} }. All the experiments below are performed on a machine of with a 32-core 2.0GHz Intel
Xeon E5-2600 Processor and 64GB of RAM. The benchmarks and full tables of experiment statistics are also available on the tool page. 
\begin{table*}[t!h!]
  \begin{center}
\begin{tabular}{|l||r|r|r||r|r|r|r|r|}
\hline
ID & \#Var & \#Arith & Nonlinear & Time$_{\text{S}}$ & Proof Size & \#Sub & \#Axiom & Time$_{\text{PC}}$ \\
\hline
\hline
461 & 6 & 36 & poly & 1.740 & 2145155 & 2 & 17442 & 203.886 \\
\hline
789 & 6 & 86 & atan2,sqrt & 1.640 & 350329 & 2 & 2464 & 128.077 \\
\hline
792 & 6 & 828 & atan2,sqrt & 0.400 & 19837 & 2 & 118 & 113.004 \\
\hline
745 & 6 & 36 & poly & 0.750 & 677580 & 2 & 5222 & 59.865 \\
\hline
785 & 6 & 80 & atan2,sqrt & 0.470 & 63388 & 2 & 526 & 26.450 \\
\hline
760 & 6 & 2767 & atan2,sqrt & 0.140 & 711 & 2 & 5 & 21.089 \\
\hline
820 & 6 & 95 & atan2,sqrt & 0.080 & 9134 & 2 & 54 & 14.703 \\
\hline
815 & 6 & 95 & atan2,sqrt & 0.330 & 41954 & 2 & 279 & 14.703 \\
\hline
814 & 6 & 95 & atan2,sqrt & 0.350 & 42102 & 2 & 278 & 14.703 \\
\hline
816 & 6 & 96 & atan2,sqrt & 0.110 & 12195 & 2 & 92 & 4.994 \\
\hline
817 & 6 & 96 & atan2,sqrt & 0.090 & 11792 & 2 & 93 & 4.993 \\
\hline
784 & 6 & 80 & atan2,sqrt & 0.060 & 7203 & 2 & 56 & 3.595 \\
\hline
781 & 6 & 86 & atan2,sqrt & 0.060 & 7481 & 2 & 45 & 2.657 \\
\hline
793 & 6 & 834 & atan2,sqrt & 0.020 & 18 & 1 & 1 & 1.855 \\
\hline
796 & 6 & 834 & atan2,sqrt & 0.010 & 18 & 1 & 1 & 1.710 \\
\hline
752 & 6 & 17 &  poly & 0.080 & 46360 & 2 & 277 & 1.709 \\
\hline
783 & 6 & 825 &atan2,sqrt & 0.020 & 93 & 1 & 1 & 1.549 \\
\hline
779 & 6 & 201 & atan2,sqrt & 0.010 & 10 & 1 & 1 & 0.705 \\
\hline
867 & 6 & 17 &  poly & 0.040 & 25820 & 2 & 147 & 0.683 \\
\hline
742 & 6 & 55 & acos,atan2,sqrt & 0.001 & 7 & 1 & 1 & 0.299 \\
\hline
508 & 6 & 53 & acos,sqrt & 0.001 & 8 & 1 & 1 & 0.286 \\
\hline
507 & 6 & 29 & acos,sqrt & 0.001 & 8 & 1 & 1 & 0.278 \\
\hline
744 & 6 & 24 & asin,cos,sin & 0.001 & 8 & 1 & 1 & 0.275\\
\hline
\end{tabular}
  \end{center}
  \caption{
    Experimental results (Proved instances):
    ID = Problem ID,
    \#Var = Number of variables,
    \#Arith = Number of arithmetic operators,
    Nonlinear = Nonlinear operators occurred in problem,
    Proof Size = Number of lines of the proof,
    $\mathrm{TIME_S}$ = Solving time in seconds,
    \#Sub = Number of subproblems generated by proof checking,
    \#Axiom = Number of proved axioms,
    $\mathrm{TIME_{PC}}$ = Proof-checking time in seconds.
  }\label{tbl:exp}
\end{table*}


\begin{table*}[t!h!]
  \begin{center}
\begin{tabular}{|l||r|r|r||r|r|r|r|r|}
\hline
ID & \#Var & \#Arith & Nonlinear & Time$_{\text{S}}$ & Proof Size & \#Sub \\
\hline
\hline
260.smt2 & 6 & 90 &  poly & 5.030 & 6281203 & 1 \\
\hline
866.smt2 & 6 & 38 & sqrt & 0.390 & 543061 & 21476 \\
\hline
775.smt2 & 6 & 2765 & atan2,sqrt & 4.040 & 130253 & 2 \\
\hline
764.smt2 & 6 & 2767 & atan2,sqrt & 1.700 & 49657 & 2 \\
\hline
762.smt2 & 6 & 2767 & atan2,sqrt & 2.040 & 42238 & 2 \\
\hline
484.smt2 & 6 & 1835 & acos,atan2,sqrt & 0.060 & 16 & 1 \\
\hline
485.smt2 & 6 & 1961 & acos,atan2,sqrt & 0.070 & 16 & 1 \\
\hline
498.smt2 & 6 & 573 & acos,matan,sqrt & 0.010 & 11 & 8191 \\
\hline
\end{tabular}
  \end{center}
  \caption{
    Experimental results (Unproved instances, Timeout = 300 sec):
    ID = Problem ID,
    \#Var = Number of variables,
    \#Arith = Number of arithmetic operators,
    Nonlinear = Nonlinear operators occurred in problem,
    Proof Size = Number of lines of the proof,
    $\mathrm{TIME_S}$ = Solving time in seconds,
    \#Sub = Number of subproblems generated by proof checking,
  }\label{tbl:exp}
\end{table*}


A main set of benchmarks that we studied is from the Flyspeck project~\cite{DBLP:conf/dagstuhl/Hales05,DBLP:conf/nfm/SolovyevH13}, which aims at a 
fully formalized proof of the Kepler conjecture. As lemmas for the
proof, hundreds of nonlinear real inequalities need to be verified.
Although the formulas usually contain only around ten variables, they
contain a huge number of nonlinear arithmetic operations and
trigonometric functions, and are mathematically challenging. In the original proof, Hales implemented procedures that combine
linear programming and interval arithmetic to establish all these
formulas, but the algorithms are formally verify. In
fact, the formal verification of these nonlinear inequalities is the last main piece
of work needed to complete the full project. 
Without any particular optimization on ICP, we have observed
promising results. Out of 916 nonlinear formulas in the Flyspeck
project repository, the solver returns {\sf unsat} for 107 of them with a timeout of 5 minute each, and a precision $\delta=10^{-3}$. Out of these formulas, we automatically generated and validated the proofs for 72 instances. The proof traces of these formulas can be very
large; for instance, we proved one with more than 2M lines in the
proof (54MB file). In Table~\ref{tbl:exp}, we list some of the
representative benchmarks to show scalability. Many of these formulas are highly nonlinear, for instance the formula encoded in 760.smt2 is following one 
\begin{multline*}
\forall\vec{x} \in [4.0, 6.3504]^5\; \Big(2\mathrm{arctan} (\frac{\Delta_2(\vec x)}{\sqrt{\Delta_1(\vec x) + \Delta_2(\vec x)^2} + \sqrt{\Delta_1(\vec x)}})\\
- 0.458(\sqrt{x_2} + \sqrt{x_3} +\sqrt{x_4} + \sqrt{x_5}) + 0.342\sqrt{x_1} + 3.319204\Big) < 0.0
\end{multline*}
where
\begin{eqnarray*}
  \Delta_1(\vec{x}) &=& 4x_1 (8x_1 (-x_1 + x_2 + x_3 + x_4 + x_5 - 8)\\
   & &+ x_2 x_5 (x_1 - x_2 + x_3 + x_4 - x_5+8\\
& &+ x_3x_4(x_1 + x_2 - x_3 - x_4 + x_5 + 8)+ 8 x_2 x_3 \\
& &- x_1 x_3 x_5 - x_1  x_2  x_4 - 8 x_4 x_5))\\
\Delta_2(\vec{x}) &=& x_2 x_5 -x_2 x_3 + x_3x_4 - x_4 x_5 +x_1^2 -x_1x_2\\
& & - x_1x_3 - x_1x_4 -x_1 x_5
\end{eqnarray*}

On the other hand, as mentioned above, we fail to establish about the proofs of unsatisfiability of about 30 instances. Table 2 shows some of these instances. They typically generate proofs that are large in size, or that the branch-and-prove loop has to generate too many sub-instances such that the proof checking can not terminate. 


\section{Conclusion}

We presented our approach for extracting formal proofs from a numerically-driven
decision procedure in the DPLL$\langle$ICP$\rangle$ framework. We formalized the ICP algorithm, and showed how to validate proof trees from the unsat answers.  A main focus for our tool is to prove
nonlinear lemmas in the Flyspeck project, and we have observed promising
experimental results. We believe the approach can be combined with existing
symbolic methods, and is a first step towards a framework that bridges the gap between symbolic and numerical approaches. Further work
would involve formalization of numerical algorithms, proof abstractions, local heuristics, and an implementation of our proof checker in standard proof assistants.

\bibliographystyle{abbrv}
\bibliography{ref_cade}

\end{document}